\documentclass[conference]{IEEEtran}
\IEEEoverridecommandlockouts
\usepackage{cite}
\usepackage{amsmath,amssymb,amsfonts}
\usepackage{algorithmic}
\usepackage{graphicx}
\usepackage{textcomp}
\usepackage{xcolor}
\usepackage{amsthm}
\usepackage{physics}
\usepackage[colorlinks=true,urlcolor=teal,
            linkcolor=teal,citecolor=teal]{hyperref}
\usepackage{xspace}
\usepackage{dsfont}
\usepackage{csquotes}
\usepackage{listings}
\usepackage{xcolor}
\usepackage[T1]{fontenc}
\usepackage[lining]{FiraMono}
\usepackage{subcaption}
\usepackage{tikz}
\usetikzlibrary{positioning,fit,arrows.meta,calc,shadings,shapes.multipart, backgrounds}
\usepackage{layouts}
\usepackage{transparent}
\usepackage{dashrule}
\usepackage{placeins}

\usepackage[british]{babel}

\usepackage{censor}
\newboolean{anonymous}
\setboolean{anonymous}{false}
\ifbool{anonymous}{}{\renewcommand{\censor}[1]{#1}}
\ifbool{anonymous}{}{\renewcommand{\blackout}[1]{#1}}
\ifbool{anonymous}{\newcommand{\genemail}[2]{\href{mailto:xxx.xxx@xx.xx}{\blackout{#2}}}}{\newcommand{\genemail}[2]{\href{#1}{#2}}}

\newcommand{\eg}{\emph{e.g.},\xspace}

\newcommand{\etal}{\emph{et al.}\xspace}

\renewcommand{\bold}[1]{\boldsymbol{#1}}

\let\oldtexttt\texttt
\renewcommand{\texttt}[1]{{\small\oldtexttt{#1}}}

\definecolor{lfd1}{HTML}{FFFFFF} 
\definecolor{lfd2}{HTML}{E69F00}
\definecolor{lfd3}{HTML}{999999}
\definecolor{lfd4}{HTML}{009371}

\definecolor{lfd5}{HTML}{BEAED4}
\definecolor{lfd6}{HTML}{ED665A}
\definecolor{lfd7}{HTML}{1F78B4}

\definecolor{lfd2light}{HTML}{e7d6b3}
\definecolor{lfd3light}{HTML}{cecece}
\definecolor{lfd4light}{HTML}{acd0c9}
\definecolor{lfd5light}{HTML}{cbbddd}
\definecolor{lfd6light}{HTML}{e79b93}
\definecolor{lfd7light}{HTML}{c4dbe9}

\definecolor{lfd2lighter}{HTML}{f5e4bb}
\definecolor{lfd3lighter}{HTML}{e8e8e8}
\definecolor{lfd4lighter}{HTML}{d3f4ed}
\definecolor{lfd5lighter}{HTML}{dcd1f2}
\definecolor{lfd6lighter}{HTML}{f9d3d0}
\definecolor{lfd7lighter}{HTML}{def2ff}
\definecolor{midgrey}{HTML}{939393}

\definecolor{desat}{HTML}{aaaaaa}
\definecolor{darkgrey}{HTML}{666666}

\colorlet{constraint}{lfd3lighter}
\colorlet{likelysuited}{lfd4}

\newcommand{\Code}{\mathcal{C}}

\newcommand{\framework}{DQI-Kit\xspace}

\newtheorem{fact}{Fact}
\newtheorem{definition}{Definition}
\newtheorem{theorem}{Theorem}

\def\BibTeX{{\rm B\kern-.05em{\sc i\kern-.025em b}\kern-.08em
    T\kern-.1667em\lower.7ex\hbox{E}\kern-.125emX}}

\begin{document}

\bstctlcite{BSTcontrol}

\title{
From Constraint to Code:
\framework~--~A Software Framework for Decoded Quantum Interferometry 
}

\author{
\IEEEauthorblockN{\blackout{Simon Thelen}}
\IEEEauthorblockA{
\textit{\blackout{Technical University of}} \\
\textit{\blackout{Applied Sciences Regensburg}} \\
\textit{\blackout{University of the Bundeswehr Munich}} \\
\blackout{Regensburg/Munich, Germany} \\
\genemail{mailto:simon.thelen@othr.de}{simon.thelen@othr.de}
}
\and
\IEEEauthorblockN{\blackout{Wolfgang Mauerer}}
\IEEEauthorblockA{
\textit{\blackout{Technical University of}} \\
\textit{\blackout{Applied Sciences Regensburg}} \\
\textit{\blackout{Siemens AG, Technology}} \\
\blackout{Regensburg/Munich, Germany} \\
\genemail{mailto:wolfgang.mauerer@othr.de}{wolfgang.mauerer@othr.de}
}
}

\maketitle

\begin{abstract}
Trying to solve hard optimisation problems with quantum techniques requires transformations of domain objectives and constraints into formats compatible with a chosen quantum algorithm. This often introduces inefficiencies and overheads that limit or even endanger potential quantum advantage for current and future approaches.
To understand and mitigate these inefficiencies, software toolchains are essential for implementing transformations, analysing overheads and eventually selecting optimal transformation paths.
Decoded Quantum Interferometry~(DQI) is a novel approach that achieves apparent quantum advantage for certain algebraic optimisation problems.
It natively operates on Max-LINSAT, which is unusual for combinatorial
optimisation, and creates the need for software solutions that alleviate the burden of manually transforming problems of interest into this format.

We present \framework, a software framework that provides a unified, extensible interface for automatically encoding constrained optimisation problems into Max-LINSAT.
Users can describe the various types of objectives and constraints that are common in industrial optimisation problems.
Our framework converts these into Max-LINSAT instances via a series of problem transformations and computes an estimate of the expected performance of DQI on these instances.
We provide an initial analysis of the implemented transformations, discussing inefficiencies and ways to mitigate them.
\framework is the basis for our ultimate goal of establishing a standardised framework that will enable further investigations to identify practical use cases for which quantum advantage with DQI can be achieved.

\end{abstract}


\tikzset{
    constraintarrow/.style={
        draw=lfd6,
        line cap=round,
        dash pattern=on 4pt off 2pt
    }
}
\tikzset{
    auxiliaryarrow/.style={
        draw=lfd7,
        line cap=round,
        dash pattern=on 0pt off 1.5pt
    }
}
\tikzset{
    dependentarrow/.style={
        draw=lfd2,
        line cap=round,
        dash pattern=on 4pt off 1.5pt on 0pt off 1.5pt
    }
}
\tikzset{
    defaultarrow/.style={
        draw=lfd3,
        line cap=round,
    }
}

\begin{figure*}
    \centering

\begin{tikzpicture}[
    font=\footnotesize,
    node distance=0.65cm and 1.6cm,
    box/.style={
        rounded corners=3pt,
        inner sep=3pt,
        fill=lfd3lighter,
        align=center,
        draw=darkgrey,
    },
    arrow/.style={
        ->,
        thick,
        >=Stealth,
    },
    group/.style={
        rounded corners,
        draw,
        inner ysep=10pt,       
        inner xsep=8pt,
    }
]

\pgfmathsetmacro{\dist}{0.5}

\node[box] at (0, 0) (weighted-max-xorsat) {
    \textbf{Weighted}\\
    \textbf{Max-XORSAT}\\
    {\color{darkgrey} $a \oplus b \oplus c = 1 [\times 2]$}\\
    {\color{darkgrey} $a \oplus c = 0$}
};

\node[box, right=1.9cm of weighted-max-xorsat] (max-xorsat) {
    \textbf{Max-XORSAT}\\
    {\color{darkgrey} $a \oplus b \oplus c = 1$}\\
    {\color{darkgrey} $a \oplus b + c = 1$}\\
    {\color{darkgrey} $a \oplus c = 0$}
};

\node[box, left=1.2cm of weighted-max-xorsat] (ising) {
    \textbf{Classical Hamiltonian}\\
    {\color{darkgrey} $2 Z I Z - 3 ZZZ$}
};

\node[box, below=of ising, fill=constraint] (linear-obj) {
    \textbf{Linear objective}\\
    {\color{darkgrey} $4a - c $}
};
\node[box, below=of linear-obj] (modular) {
    \textbf{Set inclusion}\\
    \textbf{constraint}\\
    {\color{darkgrey} $a - b \in \{1, 2\} \mod{6}$}
};

\node[box, left=\dist of ising] (pseudo-func) {
    \textbf{Pseudo-Boolean}\\
    \textbf{function}\\
    {\color{darkgrey} $a (b + c - bc)$}
};

\newcommand{\binarylinearshift}{0.2cm}
\node[box, left=1.2*\dist of pseudo-func, fill=constraint, yshift=-\binarylinearshift] (binary-linear-obj) {
    \textbf{Binary linear objective}\\
    {\color{darkgrey} $3 x + 2y$}
};
\node[box, above=of binary-linear-obj, fill=constraint, yshift=\binarylinearshift] (logic) {
    \textbf{Boolean constraint}\\
    {\color{darkgrey}  $x \wedge (y \vee z)$}
};

\coordinate (mid-pf-ising) at ($(pseudo-func)!0.5!(ising)$);

\node[box] at (mid-pf-ising |- logic) (low-pseudo-func) {
    \textbf{Low degree pseudo-Boolean function}\\
    {\color{darkgrey} $ab - 2bc$}
};

\node[box, left=\dist of modular, fill=constraint] (linear-ineq) {
    \textbf{Linear inequality}\\
    {\color{darkgrey} $2 a - 3 b + 4 > c$}
};
\node[box, below=of pseudo-func, fill=constraint]  (linear-eq) {
    \textbf{Linear equality}\\
    {\color{darkgrey} $4 a + b = 3 c$}
};
\node[box, fill=constraint] at (linear-eq -| logic) (pseudo-const) {
    \textbf{Binary non-linear}\\
    \textbf{constraint}\\
    {\color{darkgrey} $a (b + c) \geq 2bc$}
};

\node[box] at ([yshift=0.7cm]weighted-max-xorsat |- modular) (weighted-max-linsat) {
    \textbf{Weighted}\\
    \textbf{Set-Max-LINSAT} \\
    {\color{darkgrey} $2 a - c \in \{1, 2\} \pmod p$}\\
    {\color{darkgrey} $b + 2 c \in \{0\} \pmod p [\times 2]$}
};

\node[box] at (max-xorsat |- weighted-max-linsat) (symmetric-max-linsat) {
    \textbf{Symmetric}\\
    \textbf{Set-Max-LINSAT} \\
    {\color{darkgrey} $2 a - c \in \{1, 2\} \pmod p$}\\
    {\color{darkgrey} $b + 2 c \in \{0, 1\} \pmod p$}
};

\pgfmathsetmacro{\groupshift}{4pt}
\pgfmathsetmacro{\arrowshift}{6.7pt}
\pgfmathsetmacro{\topshift}{12pt}

\begin{scope}[on background layer]

    \draw[draw=none, fill=lfd4lighter, dash pattern=on 1pt off 1pt, rounded corners=3pt]
        ([xshift=-\arrowshift]weighted-max-linsat.west) |-
        ([yshift=\groupshift+\topshift]max-xorsat.north) -|
        ([xshift=2*\groupshift]symmetric-max-linsat.east) |-
        ([yshift=-2*\groupshift]symmetric-max-linsat.south) -|
        ([xshift=-\arrowshift]weighted-max-linsat.west);
        
    \node[anchor=north west, text=lfd4] at ([yshift=\groupshift+\topshift, xshift=-\arrowshift]weighted-max-linsat.west |- max-xorsat.north) {DQI native};
    
    \draw[thick, lfd4, rounded corners=3pt]
        ([xshift=-\arrowshift]symmetric-max-linsat.west) |-
        ([yshift=\topshift]max-xorsat.north) -|
        ([xshift=\groupshift]symmetric-max-linsat.east) |-
        ([yshift=-\groupshift]symmetric-max-linsat.south) -|
        ([xshift=-\arrowshift]symmetric-max-linsat.west);
        
    \node[anchor=north west, text=lfd4] at ([yshift=\topshift, xshift=-\arrowshift]symmetric-max-linsat.west |- max-xorsat.north) {Optimal for DQI};
    
    \draw[draw=none, fill=lfd7lighter, rounded corners=3pt, dash pattern=on 1pt off 1pt]
        ([yshift=\topshift]logic.north) -|
        ([xshift=\groupshift]ising.south -| modular.east) |-
        ([yshift=-\groupshift]linear-obj.south) -|
        ([xshift=-\arrowshift]linear-obj.west) |-
        ([yshift=-\groupshift]pseudo-func.south) -|
        ([xshift=-\arrowshift]pseudo-func.west) |-
        ([yshift=-\groupshift]logic.south) -|
        ([xshift=-\arrowshift]logic.west -| binary-linear-obj.west) |-
        ([yshift=\topshift]logic.north);
    
    \node[anchor=north west, text=lfd7] at ([yshift=\topshift, xshift=-\arrowshift]binary-linear-obj.west |- logic.north) {Objectives};
    
    \draw[draw=none, fill=lfd7lighter, rounded corners=3pt, dash pattern=on 1pt off 1pt]
        ([yshift=\topshift]binary-linear-obj.north) -|
        ([xshift=\groupshift]binary-linear-obj.east) |-
        ([yshift=1.7*\groupshift]linear-eq.north) -|
        ([xshift=2.7*\groupshift]linear-eq.east) |-
        ([yshift=\arrowshift]modular.north) -|
        ([xshift=\groupshift]modular.east) |-
        ([yshift=-\groupshift]modular.south) -|
        ([xshift=-\arrowshift]binary-linear-obj.west) |-
        ([yshift=\topshift]binary-linear-obj.north);
        
    \node[anchor=north west, text=lfd7] at ([yshift=\topshift, xshift=-\arrowshift]binary-linear-obj.west |- binary-linear-obj.north) {Constraints};
\end{scope}

\draw[-, thick, constraintarrow] ($(weighted-max-xorsat.east) + (0, 0.75pt)$) -- ($(max-xorsat.west) + (-2.5pt, 0.75pt)$);
\draw[-, thick, dependentarrow] ($(weighted-max-xorsat.east) + (0, -0.75pt)$) -- ($(max-xorsat.west) + (-2.5pt, -0.75pt)$);
\draw[-stealth, very thick, draw=lfd6, draw opacity=0] (weighted-max-xorsat.east) -- (max-xorsat.west) ;

\draw[-, thick, constraintarrow] ($(weighted-max-linsat.east) + (0, 0.75pt)$) -- ($(symmetric-max-linsat.west) + (-2.5pt, 0.75pt)$);
\draw[-, thick, dependentarrow] ($(weighted-max-linsat.east) + (0, -0.75pt)$) -- ($(symmetric-max-linsat.west) + (-2.5pt, -0.75pt)$);
\draw[-stealth, very thick, draw=lfd6, draw opacity=0] (weighted-max-linsat.east) -- (symmetric-max-linsat.west) ;

\draw[-stealth, thick, dependentarrow] (weighted-max-xorsat) -- (modular);
\draw[-stealth, thick, defaultarrow] (modular) -- (weighted-max-linsat);
\draw[-stealth, thick, defaultarrow] (ising) -- (weighted-max-xorsat);
\draw[-stealth, thick, constraintarrow] (pseudo-func) -- (ising);
\draw[-stealth, thick, auxiliaryarrow] (pseudo-func) -- (low-pseudo-func);
\draw[-stealth, thick, defaultarrow] (low-pseudo-func) -- (ising);
\draw[-stealth, thick, dependentarrow] (logic) -- (pseudo-func);
\draw[-stealth, thick, auxiliaryarrow] (pseudo-const) -- (pseudo-func);
\draw[-stealth, thick, auxiliaryarrow] (pseudo-const) -- (linear-eq);
\draw[-stealth, thick, constraintarrow] (binary-linear-obj) -- (pseudo-func);
\draw[-stealth, thick, defaultarrow] (linear-eq) -- (modular);
\draw[-stealth, thick, defaultarrow] (linear-ineq) -- (modular);
\draw[-stealth, thick, constraintarrow] (linear-obj) -- (modular);

\end{tikzpicture}
    
    \caption{
        Transformations between formulations implemented by \framework.
        \textcolor{lfd6}{Red dashed lines} \protect\begin{tikzpicture}[baseline=-0.5ex] \protect\draw[constraintarrow, very thick] (0,0) -- (11.5pt,0); \end{tikzpicture} represent transformations that increase the number of constraints;
        \textcolor{lfd7}{blue dotted lines} \protect\begin{tikzpicture}[baseline=-0.5ex] \protect\draw[auxiliaryarrow, very thick] (0,0) -- (11.5pt,0); \end{tikzpicture} represent transformations that add additional auxiliary variables;
        \textcolor{lfd2}{yellow dash-dotted lines} \protect\begin{tikzpicture}[baseline=-0.5ex] \protect\draw[dependentarrow, very thick] (0,0) -- (12.5pt,0); \end{tikzpicture} represent transformations that lead to many linearly dependent constraints often negatively impacting DQI performance;
        \textcolor{lfd3}{grey solid lines} \protect\begin{tikzpicture}[baseline=-0.5ex] \protect\draw[thick, defaultarrow, very thick] (0,0) -- (11.5pt,0); \end{tikzpicture} represent (favourable) transformations where none of the above apply.
    }
    \label{fig:transformations}
\end{figure*}
\section{Introduction}
Owing to their computational difficulty even for small instance sizes, hard combinatorial optimisation problems are considered prime candidates for quantum advantage. Consequently, they are one of the most widely studied aspects of quantum software.
Many proposed algorithmic approaches such as quantum annealing and QAOA can be applied to optimisation problems in QUBO format~(Quadratic Unconstrained Binary Optimisation)~\cite{Yarkoni2022,Blekos2024,Thelen:2025,Thelen2024,baierstadler:2021,bharti:2022}.
This, in principle, allows for  approximating solutions to all NP-complete optimisation problems, as the required transformations are well known~\cite{Lucas2014}.
Although QUBO formulations are, in practical applications, much less common than classical problem formulations such as Boolean satisfiability (SAT) or mixed-integer linear programs (MILP), extensive research has been conducted on how to bridge the gap between domain problem descriptions and QUBO formulations~\cite{Zaman2021,Lobe2023,Franco2023}.
This enables domain experts to apply QUBO-based quantum algorithms to a broad class of industrial use cases.

Decoded Quantum Interferometry (DQI) provides a new algorithmic framework that reduces combinatorial optimisation problems to the decoding problem of a classical error-correcting code~\cite{Jordan2025}.
Unlike most previous quantum approaches, which are largely heuristic in nature, the approximation performance of DQI can often be efficiently computed analytically.
This allows researchers and engineers to evaluate the potential of DQI in size regimes that would otherwise be inaccessible with current hardware, enabling them to identify problem classes suitable for the algorithm.
The optimisation problems native to DQI are Max-XORSAT and its generalisation Max-LINSAT.
Formulating domain use cases as Max-LINSAT instances works very differently than for established formulations~\cite{Schmidbauer2025,gabor:19:qtop,Krueger:2020} like SAT, MILP or QUBO.
Consequently, both research and tool support are necessary, as
both are currently lacking.
Moreover, the solution quality of DQI varies substantially between instances as it depends on properties of an error-correcting code based on the specific Max-LINSAT instance.
Thus, assessing the suitability of DQI for a specific problem formulation requires knowledge of coding theory, adding further barriers not only for many quantum researchers, but also for domain users.
Published research on potential applications of DQI has mostly focused on coding-theoretic problems that have little to no practical relevance outside of purely algebraic contexts.

Consequently, we introduce \emph{\framework}, a framework that aims to aid quantum software engineers in overcoming these difficulties.
It available in our \href{https://github.com/lfd/qsw-2026-dqi-kit}{code repository} and \href{https://doi.org/10.5281/zenodo.20233787}{reproduction package} (links in PDF)~\cite{mauerer:22:q-saner}.
\framework provides an extensible, unified interface to model various types of constraints and objectives that are common in many industrial combinatorial optimisation problems.
Our framework converts such constraints and objectives into DQI-native Max-LINSAT formulations by performing a series of problem transformations, as shown in \autoref{fig:transformations}.
\framework also estimates the approximation performance achievable with DQI and compares it to the results of several classical solvers.
As visualised by the different-coloured arrows in \autoref{fig:transformations}, not all problem encodings are equally efficient, highlighting the need for a systematic exploration of the space of possible formulations and transformation paths.
We aim to initiate this exploration by providing a first analysis of the efficiency of the encoding approaches implemented by our framework (which is known to be an important issue in general for quantum algorithms~\cite{Gogeissl:2024,bharti:2022,schoenberger:23:pvldb,Schmidbauer:2026,Schmidbauer2023,Schmidbauer2025,schmidbauer:25:qce}).

The rest of this work is structured as follows:
In \autoref{sec:preliminaries}, we explain the context of our work by presenting the Max-LINSAT problem formulation and by providing brief introductions to both error-correcting codes and the DQI algorithm.
\autoref{sec:problem-encoding} covers the problem encodings and transformations implemented by \framework.
In Sections \ref{sec:suitable-max-linsat-instances} and \ref{sec:weighted-max-linsat}, we derive guidelines on how to identify problem instances suitable to DQI and options to mitigate some of the limitations of DQI.
\autoref{sec:framework} presents the \framework framework.
Related work is discussed in \autoref{sec:related-work}.
Finally, we summarise our contributions and give an overview on future work in \autoref{sec:conclusion}.

\section{Preliminaries} \label{sec:preliminaries}
\subsection{The Max-LINSAT Problem}
The canonical problem that DQI is designed to solve is \emph{Max-LINSAT}, formalised as follows:
\begin{definition} \label{def:max-linsat}
    Given a finite field $\mathbb{F}_q$,
    a matrix $\boldsymbol{B} \in \mathbb{F}_q^{m \times n}$ and a vector $\bold{v} \in \mathbb{F}_q$, \emph{Max-LINSAT} is the problem of finding a variable assignment $\boldsymbol{x} \in \mathbb{F}_q^n$ that maximises 
    \begin{equation*}
        f(\boldsymbol{x}) = \sum_{i = 1}^m f_i(\boldsymbol{x})
        \text{ with }
        f_i(\bold{x}) = \begin{cases}
            1 &\text{if } \bold{x} \cdot \bold{b}_i = v_i, \\
            0 &\text{otherwise}
        \end{cases}
    \end{equation*}
    where $\boldsymbol{b}_i$ denotes the $i$-th row of $\boldsymbol{B}$ and $\cdot$ denotes the Euclidean inner product.
\end{definition}
In other words, a Max-LINSAT instance is defined by a set of $m$ linear constraints over $n$ variables $x_j \in \mathbb{F}_q$, where the $i$-the constraint is $\sum_{j = 1}^n B_{i j} x_j = v_i$.
The objective is then to find a variable assignment that satisfies as many constraints as possible.
Recall that for any finite field $\mathbb{F}_q$, $q$ is either a prime or a prime power.
If $q$ is prime, the field $\mathbb{F}_q$ can be defined~(up to isomorphism) as the set $\mathbb{Z}_q = \{0, 1, \dots, q - 1\}$ along with addition and multiplication modulo $q$.
Thus, for prime $q$, a Max-LINSAT instance is a system of linear equations modulo $q$ where the goal is, again, to satisfy as many equations as possible (intuitively, MAX-LinSat can be seen as the optimisation version of solving linear systems over finite fields).
The problem is NP-hard in general and independent of $q$, which
makes it an interesting target to study for quantum computing (in cases where all constraints are satisfiable simultaneously, Max-LINSAT can be solved efficiently using Gaussian elimination.
So, one typically considers the overdetermined regime for $m > n$).

Due to its similarity to SAT, the special case for $q = 2$ with $x_j \in \{0, 1\}$ has been studied more extensively in the past~(\eg Refs.~\cite{Mezard2003,Ibrahimi2012,Pittel2016}).
Since addition modulo $2$ is equivalent to the binary XOR operation, Max-LINSAT with $q = 2$ is typically referred to as \emph{Max-XORSAT}.

\subsection{Error-Correcting Codes}
Let us now briefly review basic definitions and facts from coding theory relevant to determine efficient problem transformations for DQI, as this guides the design of our software architecture. For a more thorough introduction, see, for instance, Refs.~\cite{Lin2021,Guruswami2025}.

Given some alphabet $\Sigma$, an error-correcting code of block length $n$ and dimension $k$ is a subset $\Code \subseteq \Sigma^n$ along with an function which injective maps each message $\boldsymbol{m} \in \Sigma^k$ to a codeword $\boldsymbol{c} \in \Sigma^n$ where $n > k$.
The goal is to add redundancy by spreading the information of a single message symbol over several codeword symbols. 
This makes it possible to reconstruct the original codeword and therefore the original message, even though a few symbols were corrupted during the transmission over a noisy communication channel.
Given some corrupted codeword $\tilde{\bold{c}} \in \Sigma^n$, the goal of a decoder is to find the codeword $\bold{c} \in \mathcal{C}$ with the smallest Hamming distance to $\tilde{\bold{c}}$ where the \emph{Hamming distance} is defined as the number of positions at which to symbol strings differ.
If the Hamming distance between two distinct codewords $\bold{c}, \bold{c}' \in \mathcal{C}$ is small, then a corrupted can be close to both of them, making unique decoding impossible, even with relatively few errors.
Thus, one ideally wants the Hamming distance between codewords to be large.
In fact, unique decoding of $\ell$ errors, that is $\ell$ incorrect symbols, is only possible, in the worst case, if $\ell < d_\text{min} / 2$.
Here, $d_\text{min}$ denotes the \emph{minimum distance}, which the smallest Hamming distance between distinct codewords in $\mathcal{C}$.
Once we try to decode beyond $d_\text{min} / 2$, there will always be strings for which more than one codeword is within Hamming distance $\ell$, making unique decoding impossible.
However, for some codes, this situation is rare even for $\ell \geq d_\text{min} / 2$ as long as $\ell$ is not too large.
For these codes, reliable decoding far beyond half the minimum distance is possible if we can tolerate some small probability of decoding incorrectly.

\emph{Linear codes} are a special case of error-correcting codes where the alphabet is a finite field $\mathbb{F}_q$ and messages are vectors over this field.
The most common field choice is the binary field $\mathbb{F}_2 = \{0, 1\}$.
Given a finite field $\mathbb{F}_q$, a linear code $\mathcal{C}$ forms a $k$-dimensional subspace of the vector space $\mathbb{F}_q^n$, which is defined as the kernel of the so-called \emph{parity-check matrix} $\bold{H}$.
In other words, $\mathcal{C} = \{\bold{c} \in \mathbb{F}_q^n \mid \bold{H c = 0}\}$.
The distance between codewords, and thus the decoding capability of a linear code, is completely determined by its parity-check matrix:
\begin{fact} \label{fact:linearly-dependent-columns}
    Given a linear code $\mathcal{C}$ with parity-check matrix $\boldsymbol{H}$
    and a codeword $\bold{c} \in \mathcal{C}$, then there exists a codeword $\bold{c}' \in \mathcal{C}$ at Hamming distance $d$ of $\bold{c}$ if and only if $\boldsymbol{H}$ contains a set of $d$ linearly dependent columns.
    In particular, the minimum distance of $\mathcal{C}$ is equal to the smallest number $d_\text{min}$ such that $\boldsymbol{H}$ has $d_\text{min}$ linearly dependent columns.
\end{fact}
This fact is particularly relevant in the context of DQI, as we will see in \autoref{sec:suitable-max-linsat-instances}, because it allows us to determine a priori how well DQI can approximate the optimal solution of a given Max-LINSAT instance.

The decoding problem for linear codes is NP-complete in general, even in the unique decoding regime.
In practice, however, efficient decoding is often possible via specialised decoders that exploit the algebraic structure of specific code families, or via heuristic methods such as belief propagation or information-set decoding.
Although these latter approaches do not offer worst-case guarantees, they still deliver excellent decoding performance on a wide range of codes.

\subsection{Decoded Quantum Interferometry}
Decoded Quantum Interferometry is a quantum algorithm which approximates solutions to Max-LINSAT instances.
It does so by preparing a superposition of all possible solutions that is biased toward good solutions.
Thus, when measuring this state, one likely obtains a good solution.
Concretely, DQI prepares the \emph{DQI state}:
\begin{equation} \label{eq:dqi-state}
    \ket{P(f)} = \sum_{\boldsymbol{x} \in \mathbb{F}_q^n} P(f(\boldsymbol{x})) \ket{\boldsymbol{x}}\text{.}
\end{equation}
Here, $f$ is the objective function as specified in \autoref{def:max-linsat} and $P$ is a polynomial of degree $\ell$.
The larger the value of $\ell$, the stronger the separation between good and bad solutions and the larger the probability of measuring a good solution.
However, how large one can choose $\ell$ to be depends on the error-correcting code $\mathcal{C}$ with parity-check matrix $\bold{B}^T$ where $\bold{B}^T$ is the transpose of the constraint matrix $\bold{B}$ from \autoref{def:max-linsat}.
DQI can only prepare \eqref{eq:dqi-state} exactly if there is a decoding algorithm for $\mathcal{C}$ that can perfectly correct up to $\ell$ errors.
For imperfect decoders, for example when $\ell$ is larger than half the minimum distance of $\mathcal{C}$, DQI can still approximate Max-LINSAT solutions by approximating the state \eqref{eq:dqi-state}.
However, the solution quality depends on how many error patterns with up to $\ell$ errors can be decoded correctly.

In the perfect decoding regime and also in the imperfect regime for some instance families, the expected number of satisfied constraints when measuring the DQI state can be efficiently computed classically.
This permits analysis of the algorithm on industrial-scale instances using current hardware.

In their study of DQI, Jordan~\etal~\cite{Jordan2025} consider a generalisation of Max-LINSAT, which we will refer to as \emph{Set-Max-LINSAT}.
Here, the $v_i \in \mathbb{F}_q$ are replaced by sets $F_i \subseteq \mathbb{F}_q$.
The $i$-th constraint is satisfied when $\bold{b} \cdot \bold{x} \in F_i$.
Since each such set inclusion constraint can be decomposed into disjoint equality constraints ($\bold{b} \cdot \bold{x} = v$ for all $v \in F_i$), Set-Max-LINSAT describes exactly the same optimisation problems as standard Max-LINSAT.
However, combining equality constraints with the same left-hand side, as described by the $\bold{b}_i$, into a single set inclusion constraint not only reduces the number of constraints but can also improve the approximation performance of DQI for reasons described below in \autoref{sec:suitable-max-linsat-instances}.

\section{Problem Encoding for Max-LINSAT} \label{sec:problem-encoding}

We now commence by outlining Max-LINSAT formulations for typical objective and constraint types, as they are found in many industrial combinatorial optimisation problems.
These formulations are implemented as part of the \framework software framework.
Owing to the multitude of possible encoding approaches proposed in the abundant optimisation literature, we necessarily need to restrict our efforts to the most common input formats.
Our extensible software design allows, however, to easily add further encodings, and we anticipate the selection will grow in the future.
We discuss some alternative encoding techniques alongside the implemented approaches.

Like QUBOs or other unconstrained problem formulations, Max-LINSAT does not support hard constraints: Every constraint can be broken,
and solution quality is measured by the number of broken constraints, which is reasonable for optimisation. However, most industrial optimisation problems include at least some hard constraints that immediately result in invalid solutions when 
they are broken.
To express these in the Max-LINSAT formalism, we treat them as soft constraints that result in a sufficiently large penalty if broken.
This makes it useful to consider a weighted version of Max-LINSAT for encoding constraints:
\begin{definition} \label{def:weighted-max-linsat}
    Given a finite field $\mathbb{F}_q$, a matrix $\bold{B} \in \mathbb{F}_q^{m \times n}$, a list of sets $F_1, \dots, F_m$ with $F_i \subseteq \mathbb{F}_q^m$ and a weight vector $\bold{w} \in \mathbb{Q}^m$ with $w_i > 0$,
    \emph{Weighted Set-Max-LINSAT} is the problem of finding a vector $\bold{x} \in \mathbb{F}_q^n$ that maximises
    \begin{equation*}
        f(\bold{x}) = \sum_{i = 1}^m f_i(\bold{x})
        \text{ with }
        f_i(\bold{x}) = \begin{cases}
            w_i &\text{if } \bold{x} \cdot \bold{b}_i \in F_i \\
            0 &\text{otherwise.}
        \end{cases}
    \end{equation*}
\end{definition}
We  discuss ways of dealing with these weights in the context of DQI in \autoref{sec:weighted-max-linsat}. While every weighted Set-Max-LINSAT instance can, in principle, can be converted into an unweighted instance by duplicating constraints, this is usually far from optimal.

\subsection{Equality and Inequality Constraints} \label{sec:equality-inequality-constraints}
DQI operates over a finite field $\mathbb{F}_q$, meaning each variable has $q$ possible values if not otherwise restricted.
This field structure allows for a natural encoding of many types of constraints.
For example, one type of constraint directly encodable into Max-LINSAT is an equality constraint $x_i = x_j$, which is expressed as $x_i + (-x_j) = 0$ where $-x_j$ is the additive inverse of $x_j$ in $\mathbb{F}_q$.
Beyond that, we can also easily encode inequality constraints of the form $x_i \neq x_j$. These are especially common for categorical variables, for example to describe a scheduling conflict.
A \enquote{not equal} constraint is expressed as $x_i + (-x_j) \in \mathbb{F}_q \setminus \{0\}$.
This is efficient to encode even for large $q$ since we do not have to represent the set $F_i$ for a given constraint $\bold{b}_i \cdot \bold{x} \in F_i$ explicitly during the DQI algorithm.
In DQI, each $F_i$ is encoded as is the Fourier transform of a shifted and scaled version of the phase oracle
\begin{equation} \label{eq:phase-oracle}
y \mapsto \begin{cases} +1 \text{ if } y \in F_i \text{,}\\ -1 \text{ if } y \not\in F_i \text{.}\end{cases}
\end{equation}
For simple sets such $F_i = \{0\}$ and $F_i = \mathbb{F}_q \setminus \{0\}$ the number of gates required to implement this phase oracle, and thus the encoding of the $F_i$ during the DQI algorithm, is only logarithmic in $q$.

Clearly, we can extend these types of equality and inequality constraints to arbitrary linear combinations of variables.
If $q$ is prime, can can interpret these as (in)equalities modulo $q$.
However, this interpretation does not hold when $q = p^k$ for a prime $p$ and $k > 1$.
This is because the field operations then correspond to addition and multiplication of degree-$k$ polynomials over $\mathbb{F}_p$,
which is algebraically distinct from the same operations over the integers modulo $q$.

Non-modular linear equalities and inequalities over the integers are critical to express many industrial optimisation problems~\cite{Nemhauser1988,Williams2013}.
To encode these, for each integer variable $x_i \in \mathbb{Z}$, we must select a range $l_i, u_i \in \mathbb{Z}$ with $x_i \in [l_i, u_i]$.
Then, we select a prime $p$ large enough so all integer constraints can be expressed modulo that prime.
Along with the actual constraints we want to encode we also need to add constraints $x_i \in \{l_i, \dots, u_i \} \pmod{p}$ for all integer variables to restrict their range, each with a sufficiently large weight.
Depending on the types of constraints we want to express, $p$ might need to be quite a bit larger then the sizes of the ranges $[l_i, u_i]$.
In order for an \enquote{equal} constraint $\bold{b}_i \cdot \bold{x} = v_i$ or a \enquote{not equal} constraint $\bold{b}_i \cdot \bold{x} \neq v_i$ to also hold modulo $p$, $p$ must be larger than $\max_{\bold{x} \in \mathcal{X}} |\bold{b}_i \cdot \bold{x} - v_i|$ where $\mathcal{X} = [l_1, u_1] \times \cdots \times [l_m, u_m]$ is the set of allowed values for the $\bold{x}$.
This is visualised in \autoref{fig:modulo-constraints-equality}.
Ordering equalities ($\leq, <, \geq, >$) are also possible.
For example $\bold{b}_i \cdot \bold{x} \geq v_i$ can be encoded as the set inclusion constraint $\bold{b}_i \cdot \bold{x} - v_i \in \{a \in \mathbb{F}_q \mid a \in \{0, \dots, u_i\}\pmod{p}\}$.
However, for this, $p$ generally has to be larger than for \enquote{equal} or \enquote{not equal} constraints.
To be precise, $p$ must be greater than $U_i - L_i$ where $L_i = \min_{\bold{x} \in \mathcal{X}} (\bold{b}_i \cdot \bold{x} - v_i)$ and $U_i = \max_{\bold{x} \in \mathcal{X}} (\bold{b}_i \cdot \bold{x} - v_i)$ as shown in \autoref{fig:modulo-constraints-inequality}.
\begin{figure}
    \centering
    \begin{subfigure}[t]{0.45\linewidth}
        \centering
        \includegraphics{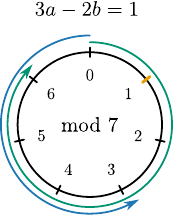}
        \caption{$=$ or $\neq$ constraint}
        \label{fig:modulo-constraints-equality}
    \end{subfigure}
    \begin{subfigure}[t]{0.45\linewidth}
        \centering
        \includegraphics{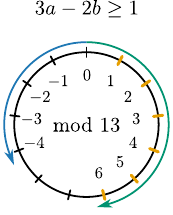}
        \caption{$\leq,<, \geq, >$ constraint}
        \label{fig:modulo-constraints-inequality}
    \end{subfigure}
    \caption{
        Max-LINSAT encodings of integer constraints:
        Ordering inequalities  generally require a larger field order.
        In the example, we assume that $a, b \in \{0, 1, 2\}$, so $-4 \leq 3a - 2 b \leq 6$.
    }
    \label{fig:modulo-constraints}
\end{figure}

If we increase the order of the field $\mathbb{F}_p$ to accommodate non-modular integer constraints, existing modular constraints modulo a smaller prime $p' < p$ need to be modified to conform to the larger field.
For example, $\bold{b}_i \cdot \bold{x} = v_i \pmod{p'}$ must be replaced with $\bold{b}_i \cdot \bold{x} \in \{a \in \mathbb{F}_p \mid a = v_i \pmod{p'}\}$.
If $p \neq p'$, then $p$ and $p'$ are coprime. Therefore, $p$ must be greater than $U_i - L_i$ for this to correctly encode the constraint, just like with ordering constraints.

Another option to encode integer constraints is to use a $q$-ary encoding of integer variables where each Max-LINSAT variable represents a single digit of the integer variable.
Using this approach, Sabater \etal~\cite{Sabater2025} developed binary-encoded formulations for integer constraints.
These are based on Max-XORSAT gadgets implementing a binary ripple-adder.
Their approach results in a $\bold{B}$ matrix with many linearly dependent sets of $3$ rows, inhibiting DQI performance, as explained in \autoref{sec:weighted-max-linsat}.
For this reason, we chose not to include their encoding technique in the first version of \framework in favour of our more direct approach, which we deem to be more efficient and broadly applicable.

\subsection{Linear Objectives} \label{sec:linear-objectives}
The goal of many optimisation problems is to find a variable assignment that maximises or minimises some function under a set of constraints.
While Max-LINSAT does not support such objective functions directly, we can still encode them when allowing weighted constraints.
To illustrate this, we can use the constraint $x_i = 1$ for $x_i \in \{0, 1\}$ with weight $w_i$ to represent the objective function $w_i x_i$, which we want to maximise.
Similarly, $x_i = 1$ with weight $w_i$ represents the maximising objective $1 - w_i x_i$ or alternatively, ignoring the constant offset of $1$, the minimising objective $w_i x_i$.
This way, each objective function that is a linear combination of binary variables can be represented by a set of weighted Max-XORSAT constraints.

Encoding linear combinations of non-binary variables is also possible but less efficient:
For example, the objective $a x_i$ with $x_i \in \{0, \dots, 6\}$ and $a \in \mathbb{Q}_+$ can be represented by three weighted constraints: $x_i \in \{4, 5, 6\}$ with weight $4 a$, $x_i \in \{2, 3, 6\}$ with weight $2 a$ and $x_i \in \{1, 3, 5\}$ with weight $a$.
These constraints can be constructed from the binary representation of each of the possible values of $x_i$, so $\lceil\log_2(u_i - l_i + 1)\rceil$ weighted constraints are necessary for a term $a x_i$ with $x_i \in \{l_i, \dots, r_i\}$.

\subsection{Polynomial Constraints and Objectives} \label{sec:qubos-pubos}
Polynomial constraints and objectives extend linear formulations by allowing interactions between multiple variables, enabling the modelling of higher-order dependencies.
Weighted Max-LINSAT can describe maximisation problems of the form $\max_{\bold{x}} P(\bold{x})$ where $P$ is some polynomial and $\bold{x} \in \{0, 1\}^k$ is vector of binary variables.
To see this, we can use the fact that Max-XORSAT instances can be interpreted as Hamiltonians that are an unweighted sum of Pauli-$Z$-Hamiltonians.
This equivalence was noticed by \cite{Jordan2025,Schmidhuber2025,Parekh2025} among others.
In light of Weighted Max-XORSAT, this gives rise to the following generalised equivalence:
\begin{fact} \label{fact:weighted-max-xorsat-hamiltonian}
    Given a Weighted Max-XORSAT instance as defined in \autoref{def:weighted-max-linsat}, its objective $f$ is equal to $f(\bold{x}) = \langle \bold{x}| H |\bold{x}\rangle$ for the Hamiltonian
    \begin{gather*}
        H = \frac{W}{2} I^{\otimes n} + \frac{1}{2}\sum_{i = 1}^m w_i (-1)^{v_i} \prod_{\substack{j = 1 \\ B_{i j} = 1}}^m Z^{(j)} 
    \end{gather*}
    where 
    $W = \sum_{i = 1}^m w_i$ and
    $Z^{(i)}$ is the Pauli-$Z$ operator applied to the $i$-th qubit.
\end{fact}
This follows from $f_i(\bold{x}) = \frac{1}{2}(1 + (-1)^{\bold{b}_i \cdot \bold{x} + v_i})$ and the fact that $Z^{(i)} \ket{\bold{x}} = (-1)^{x_i} \ket{\bold{x}}$ for all $\bold{x} \in \{0, 1\}^n$

Since we can ignore the constant offset $W / 2$ when trying to find a solution maximising $f(\bold{x})$, we can, vice versa, convert every Hamiltonian that is diagonal in the $Z$ basis into an equivalent Weighted Max-XORSAT instance.
This applies, in particular, to every Ising Hamiltonian and, thus, to every QUBO.
Encoding industrial problems as QUBOs is well-established topic in quantum software engineering~\cite{Lucas2014,Yarkoni2022}.
The corresponding problem formulations can all, in principle, be directly transformed into Weighted Max-LINSAT instances.
By generalising the QUBO reduction to Hamiltonians with Pauli-$Z$ terms of any degree, every polynomial unconstrained binary optimisation (PUBO) problem can be encoded as a Weighted Max-XORSAT instance: Formally,
a PUBO is an unconstrained optimisation problem whose objective function $\{0, 1\}^n \rightarrow \mathbb{R}$ is given by a multi-variate polynomial of arbitrary degree.
In contrast, QUBOs encode polynomials only of degree $2$.
Every function $b: \{0, 1\}^n \rightarrow \mathbb{R}$ can be described as such a multi-variate polynomial~\cite{Hammer1968}.
Since $x^2 = x$ for $x \in \{0, 1\}$, this polynomial is always multi-linear, meaning that the exponent of every variable is at most 1 of each of its monomials.
Via its representation as a multi-linear polynomial, we can encode every pseudo-Boolean function $b : \bold{x} \mapsto b(\bold{x})$ as a Weighted Max-XORSAT instance.
To do this, we simply substitute each $x_i$ in the polynomial representation of $b$ with $(1 - Z^{(i)}) / 2$ and apply the equivalence from \autoref{fact:weighted-max-xorsat-hamiltonian}.
However, since a multi-linear polynomial with $n$ variables, in general, consists of $2^n$ monomials, exponentially many Max-XORSAT constraints are necessary in the worst case.
In fact, a simple product $x_1 x_2 \dots x_k$ represented by the Hamiltonian $1/2^k (1 - Z^{(1)}) \dots (1 - Z^{(k)})$ already leads to an exponential number of constraints in $k$.

There are multiple ways to handle the exponential blow up: 
For \framework, we selected a relatively simple technique of introducing additional auxiliary variables that replace high-order sub-terms.
We then add constraints to ensure that the auxiliary variables are equal to the terms they represent.
To guarantee equality between two expressions $e_1(\bold{x})$ and $e_2(\bold{x})$, we add a weighted constraint $-(e_1(\bold{x}) - e_2(\bold{x}))^2$, which is maximised when $e_1(\bold{x}) = e_2(\bold{x})$.
For example, $x_1 x_2 x_3 x_4$ can be split into $x_1 x_2 = a_1$, $x_3 x_4 = a_2$ and $a_1 a_2$ where the two equalities are represented by the terms $-(x_1 x_2 - a_1)^2$ and $-(x_3 x_4 - a_2)^2$.
With this approach, the number of Max-XORSAT constraints only grows linearly with the number of variables at the expense of using a potentially linear number of auxiliary variables.
In addition to polynomial objectives, this auxiliary variable approach also allows us to express polynomial constraints by first replacing the term with an auxiliary variable and then encoding the constraint, as shown in \autoref{sec:equality-inequality-constraints}.
A more general option to reduce the number of constraints would be to employ a quadratisation algorithm.
This similarly reduces monomial degrees at the expense of additional auxiliary variables.
Several well-understood quadratisation techniques exist that allow for efficient implementations~\cite{Dattani2019,Schmidbauer:2026}.

Unfortunately, there does not seem to be an obvious way to generalise the encodings of polynomial objectives and constraints explained above to non-binary variables in a way that is compatible with Max-LINSAT.
These encodings exploit the group homomorphism $\mathbb{Z}_2 \rightarrow \mathbb{Z}, x \mapsto (-1)^x$.
The generalisation of this for $p > 2$ is $\mathbb{Z}_p \rightarrow \mathbb{C}, x \mapsto \mathrm{e}^{2\pi \mathrm{i} x / p}$.
Encoding this into Max-LINSAT is theoretically possible, but would result in complex constraint weights, which are not well-defined.
It seems likely that, similarly to MILP formulations, general polynomial constraints and objectives can only be expressed approximately in Max-LINSAT~\cite{Nemhauser1988,Williams2013}.

\subsection{Boolean constraints} \label{sec:logical-constraints}
The final constraint type we consider is Boolean constraints.
Every Boolean constraint, described by a Boolean function $b: \{0, 1\}^n \rightarrow \{0, 1\}$, can be transformed into a set of weighted Max-XORSAT constraints by representing $b$ as multi-linear polynomial and then applying the techniques explained in \autoref{sec:qubos-pubos}.
For example, logical AND is described by the polynomial $(x_1, x_2) \mapsto x_1 x_2$, which is transformed into the Hamiltonian $H = \frac{1}{4} (1 - Z^{(1)}) (1 - Z^{(2)})$.
This, in turn, is converted into the following Max-LINSAT constraints
\begin{align} \label{eq:and}
    x_1 &= 1 \pmod 2 \nonumber \\
    x_2 &= 1 \pmod 2 \\
    x_1 + x_2 &= 0 \pmod 2\text{,} \nonumber
\end{align}
each with weight $1/2$.
If both $x_1$ and $x_2$ are one, then all $3$ constraints are satisfied.
Otherwise, only $1$ is satisfied.
When including the constant offset of $-1/2$, this leads to an objective function, which has value one if $x_1 \wedge x_2$ and zero otherwise.
Generalising \eqref{eq:and} to $n$ variables leads to the constraints
\begin{gather*}
    \left\{
    \boldsymbol{x} \cdot \boldsymbol{b} = |\boldsymbol{b}| \,\,(\mathrm{mod}\, 2),
    \mid
    \boldsymbol{b} \in \{0, 1\}^n, 
    |\boldsymbol{b}| > 0
    \right\}
\end{gather*}
Here, all $2^n - 1$ constraints are satisfied if $|\boldsymbol{x}| = n$ and $2^{n - 1} - 1$ are satisfied otherwise.
As can be seen by this example, transforming Boolean functions into Max-XORSAT constraints can and often does lead to the exponential constraint increase described above.
The same effect, for example, happens for an $n$-variable OR constraint, which is transformed into $2^n - 1$ constraints of the form $\bold{b \cdot x} = 1 \pmod{2}$.
This necessitates the usage of degree reduction techniques as explained in \autoref{sec:qubos-pubos}.

\section{Identifying Problems Suited for DQI} \label{sec:suitable-max-linsat-instances}
Although, as shown in \autoref{sec:problem-encoding}, predominant objective and constraint types can be transformed into Max-LINSAT, not all formulations are equally suited to the DQI algorithm.
Finding formulations that work well with DQI requires identifying Max-LINSAT instances where $\Code = \{\bold{y} \mid \bold{B}^T \bold{y} = \bold{0}\}$ is a strong error-correcting code.
If we only consider perfect decoders, the relevant quantity for the quality of a code is the minimum distance of $\Code$, which, by \autoref{fact:linearly-dependent-columns}, is the smallest number $d$ such that $\bold{B}^T$ linearly dependent columns.
A set of linearly dependent columns in $\bold{B}^T$ one-to-one corresponds to a set of linearly dependent rows in $\bold{B}$.
As an illustrative example, consider the following three constraints:
\begin{align} \label{eq:constraint-cycle}
    x_1 + 2x_2 &= 2 \pmod{3} \nonumber \\
    2x_2 + x_3 &= 1 \pmod{3} \\
    x_3 + 2x_1 &= 1 \pmod{3} \nonumber
\end{align}
The linear combination $\bold{b}_1 + 2 \bold{b}_2 + \bold{b}_3$ yields
\begin{equation*}
(1, 2, 0) + 2 (0, 2, 1) + (2, 0, 1) = (0, 0, 0) \pmod{3}\text{.}
\end{equation*}
This means that the linear code corresponding to any Max-LIN\-SAT instance containing these constraints has minimum distance at most 3, independent of the right-hand sides of the constraints.
For imperfect decoders, we cannot solely rely on the minimum distance but also need to consider the distance distribution of the code.
By \autoref{fact:linearly-dependent-columns}, we want a code where most codeword pairs have a large distance.
This, again, corresponds to a problem matrix $\bold{B}$ with few linearly dependent rows.

Based on this insight, we can identify specific types of constraints that are likely to impair DQI performance:
\begin{enumerate}
    \item \emph{Duplicate constraints:}
    Two constraints $\bold{b}_i \cdot \bold{x} = v_i$ and $\bold{b}_j \cdot \bold{x} = v_j$ with $\bold{b}_i = \bold{b}_j$ immediately lead to a minimum distance of $2$, even if $v_i \neq v_j$.
    This results in a code that, in the worst case, cannot even correct a single error, since we can not distinguish whether an error occurred at position $i$ or $j$.
    This severely limits decoders, even in the imperfect decoding regime, likely degrading the approximation performance of DQI significantly.
    
    \item \emph{AND/OR constraints:}
    Examining the Boolean AND formulation in \eqref{eq:and}, we observe that the three Max-XORSAT constraints are linearly dependent resulting in a minimum distance of at most $3$.
    The same effect can also be observed for OR constraints.

    \item \emph{Short (in)equality cycles:}
    A cycle of three inequality constraints $x_1 \neq x_2$, $x_2 \neq x_3$, $x_3 \neq x_1$ leads to a situation similar to the outlined in \eqref{eq:constraint-cycle}, which, in turn, leads to a minimum distance of at most $3$.
    One example of this effect in a practical use case is the Max-cut problem where, given an undirected graph $G = (V, E)$, the objective is to colour each vertex with one of two colours in such a way that the number of edges connecting two vertices of opposite colour is maximised.
    Max-cut can naturally be encoded as a Max-XORSAT problem by using one variable $x_v \in \mathbb{F}_2$ for each vertex $v \in V$ and adding one constraint $x_u \neq x_v$ for each edge $(u, v) \in E$.
    It is easy to see that, for this problem formulation, the minimum distance of the code $\Code$ is equal to the girth of $G$, which is defined as the length of the longest cycle.
    More generally, any cycle
    \begin{gather*}
        l_1 x_1 \sim_1 r_2 x_2, \,\,\,\,
        l_2 x_2 \sim_2 r_3 x_3, \,\,\,\,
        \dots, \,\,\,\,
        l_k x_k \sim_2 r_1 x_1
    \end{gather*}
    for non-zero coefficients $l_i, r_i$ and (in)equality relations $\sim_i$ leads to a linear dependency of $k$ constraints.
\end{enumerate}

As these examples show, linear dependencies arise naturally in many practical problem formulations, making it difficult to avoid them entirely.
Therefore, in addition to identifying use cases with few inherent dependencies, we deem dependency-reduction techniques to be a promising approach for broadening the applicability of DQI.
Below, we demonstrate that mitigating linear dependencies is possible to a certain extent by outlining two concrete techniques.
Our aim is for \framework to facilitate further investigation and refinement of such techniques within the research community.

First, consider the issue of an AND constraint resulting in three linearly dependent Max-LINSAT constraints.
This linear dependency can be eliminated by encoding the Boolean constraint only approximately in the sense that the \emph{yes} case corresponds to the optimal variable assignment, while the \emph{no} case corresponds to a sub-optimal but not necessarily worst-possible assignment.
For example, if we remove the constraint $x + y = 0 \pmod{2}$ removing the linear dependency, the set of constraints is still maximally satisfied (two out of two) when $x \wedge y$.
This means that the optimal solution of the Max-LINSAT instance is not affected.
As shown by Sabater \etal~\cite{Sabater2025}, this approach also works for the three-bit majority gate.
It is likely that it can be generalised to various other logical constraints, which could considerably improve DQI's performance on broad classes of problems.

A second, more widely applicable technique involves introducing auxiliary variables to resolve linear dependencies.
To illustrate this technique, we present a simple gadget derived from the following result:
\begin{theorem} \label{thm:gadget}
    Let $f(\bold{x})$ be the objective function of a Max-LINSAT instance with coefficient matrix $\bold{B} \in \mathbb{F}_q^{m \times n}$ and sets $F_1, \dots, F_m \subseteq \mathbb{F}_q$.
    Fix an index $i \in \{1, \dots, m\}$ and a positive integer $k$.
    Construct a new Max-LINSAT instance by 
    introducing new variables $y_1, \dots, y_k$, 
    replacing the constraint $\bold{b}_i \cdot \bold{x} \in F_i$
    with $\bold{b}_i \cdot \bold{x} + \sum_{j = 1}^k y_j \in F_i$
    and adding constraints $y_j \in \{0\}$ for all $j \in \{1, \dots, k\}$.
    Let $\hat{f}(\bold{x}, \bold{y})$ denote the objective function for this newly constructed instance.
    Then, for every $\bold{x} \in \mathbb{F}_q^n$,
    $k + f(\bold{x}) = \max_{\bold{y \in \mathbb{F}_q^k}} \hat{f}(\bold{x, y})$.
\end{theorem}
\begin{proof}
    For a given $\bold{x}$, define $v = \bold{b}_i\cdot\bold{x}$.
    If $v \in F_i$, then by setting $y_j = 0$ for $1 \leq j \leq k$, all $k + 1$ new constraints are satisfied, so $k + f(\bold{x})$ constraints are satisfiable in total.
    
    Otherwise, if $v \not\in F_i$, we can either try to satisfy $\bold{b}_i \cdot \bold{x} + \sum_{j} y_j \in F_i$, which requires us to set at least one of the $y_j$ to a non-zero value and satisfying at most $k$ out of the $k + 1$ added constraints,
    or we can satisfy all constraints of the form $y_j = \{0\}$, which leaves $\bold{b}_i \cdot \bold{x} + \sum_{j} y_j \in F_i$ unsatisfied, again satisfying $k$ out of $k + 1$ new constraints.
    In either case, $k + f(\bold{x})$ constraints are satisfiable in total.
\end{proof}
We can remove linear constraint dependencies by applying \autoref{thm:gadget}:
Let $\bold{b}_1, \dots, \bold{b}_d$ be a minimal set of linearly dependent rows, in the sense that after removing any single row, the remaining rows are linearly independent.
Then, applying the gadget from \autoref{thm:gadget} to any of the $\bold{b}_i$ ($1 \leq i \leq d$) removes the linear dependency but adds a new linear dependency of length $d + k$, which includes the $k$ added constraints $y_j \in \{0\}$.
This increases the minimum distance by $k$ at the expense of adding $k$ variables and $k$ constraints.
The gadget thus can help improve the approximation performance of DQI by reducing few small linear dependencies.
However, introducing many gadgets with large $k$ across the Max-LINSAT instance can quickly lead to diminishing returns.
This is because, for large instances, the performance of DQI depends on the ratio $\ell / m$~\cite{Jordan2025} and our gadget increases both $\ell$ and $m$ by $k$.

In this section, we considered the theoretical decodability of codes derived from Max-LINSAT instances, allowing us to derive universally applicable upper bounds on DQI performance.
However, actual performance and runtime of DQI vary depending on the classical decoder, which warrants a more thorough analysis.
For instance, belief propagation, the most widely used general purpose decoding algorithm, requires consideration of additional instance properties, such as the sparsity of the constraint matrix and the number of variable pairs shared across constraints \cite{Lin2021}.
Conversely, specialised decoders tailored to certain problem structures can improve DQI performance considerably \cite{Jordan2025,Gu2025,Khattar2025}.
For this reason, \framework is designed from the outset to support different classical decoders, enabling investigations into the effects of different decoder choices and allowing for a more comprehensive analysis of DQI performance in the future.

\newcommand{\dist}{0.4cm}
\tikzset{
    umldiag/.style={
        node distance=\dist and \dist,
        box/.style={
            rounded corners=3pt,
            inner sep=3pt,
            fill=white,
            draw=darkgrey,
            align=left,
            font=\ttfamily,
            execute at begin node=\setlength{\baselineskip}{9pt}
        }
    },
}
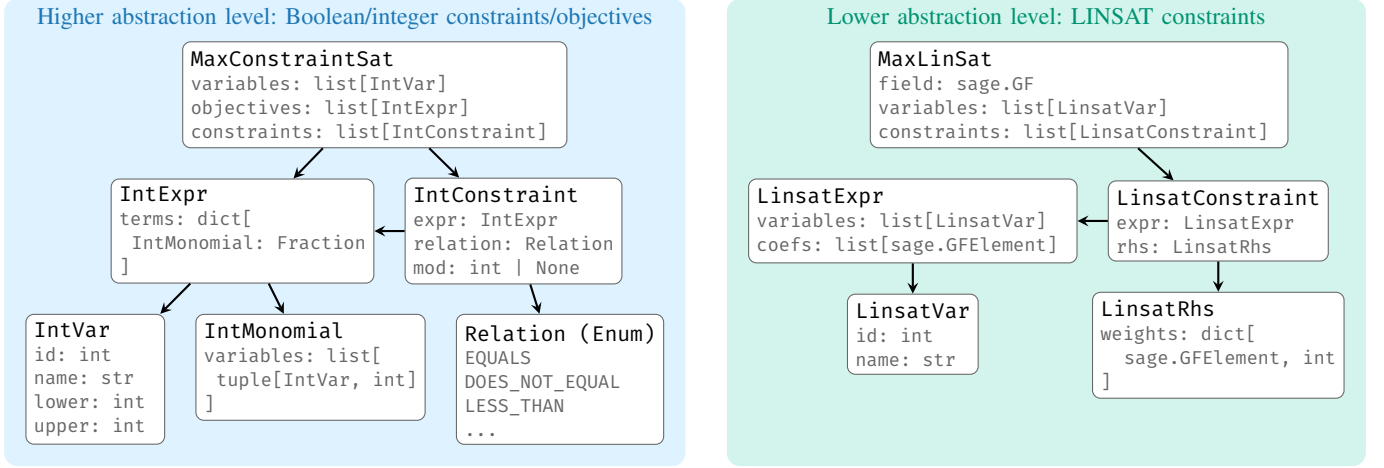
\begin{figure*}
    \centering
        \begin{tikzpicture}[umldiag, baseline=(current bounding box.north)]
            \node[box] (maxconstraintsat) {
            \footnotesize MaxConstraintSat \\
            \scriptsize\color{darkgrey} variables: list[IntVar]\\
            \scriptsize\color{darkgrey} objectives: list[IntExpr]\\
            \scriptsize\color{darkgrey} constraints: list[IntConstraint]
            };
            \node[box, below=of maxconstraintsat, xshift=-1.75cm] (intexpr) {
            \footnotesize IntExpr \\
            \scriptsize\color{darkgrey} terms: dict[\\
            \scriptsize\color{darkgrey} \,\,\,\,IntMonomial: Fraction\\
            \scriptsize\color{darkgrey} ]
            };
            \node[box, right=of intexpr] (intconstraint) {
            \footnotesize IntConstraint \\
            \scriptsize\color{darkgrey} expr: IntExpr\\
            \scriptsize\color{darkgrey} relation: Relation\\
            \scriptsize\color{darkgrey} mod: int | None
            };
            \node[box, below=of intexpr, xshift=0.9cm] (operator) {
            \footnotesize IntMonomial \\
            \scriptsize\color{darkgrey} variables: list[ \\
            \scriptsize\color{darkgrey}\,\,\,\,tuple[IntVar, int] \\
            \scriptsize\color{darkgrey}]
            };
            \node[box, anchor=north east, xshift=-\dist] at (operator.west |- operator.north) (intvar) {
            \footnotesize IntVar \\
            \scriptsize\color{darkgrey} id: int\\
            \scriptsize\color{darkgrey} name: str\\
            \scriptsize\color{darkgrey} lower: int\\
            \scriptsize\color{darkgrey} upper: int
            };
            \node[box, anchor=north west, xshift=\dist] at (operator.east |- operator.north) (relation) {
            \footnotesize Relation (Enum)\\
            \scriptsize\color{darkgrey} EQUALS\\
            \scriptsize\color{darkgrey} DOES\_NOT\_EQUAL\\
            \scriptsize\color{darkgrey} LESS\_THAN\\
            \scriptsize\color{darkgrey} ...
            };
            \draw[-stealth, thick] (maxconstraintsat) -- (intexpr);
            \draw[-stealth, thick] (maxconstraintsat) -- (intconstraint);
            \draw[-stealth, thick] (intexpr) -- (intvar);
            \draw[-stealth, thick] (intexpr) -- (operator);
            \draw[-stealth, thick] (intconstraint) -- (relation);
            \draw[-stealth, thick] (intconstraint) -- (intexpr);
            \node[box, right=4cm of maxconstraintsat] (maxlinsat) {
            \footnotesize MaxLinSat \\
            \scriptsize\color{darkgrey} field: sage.GF\\
            \scriptsize\color{darkgrey} variables: list[LinsatVar]\\
            \scriptsize\color{darkgrey} constraints: list[LinsatConstraint]
            };
            \node[box, below=of maxlinsat, xshift=-2.2cm] (linsatexpr) {
            \footnotesize LinsatExpr \\
            \scriptsize\color{darkgrey} variables: list[LinsatVar]\\
            \scriptsize\color{darkgrey} coefs: list[sage.GFElement]
            };
            \node[box, right=of linsatexpr] (linsatconstraint) {
            \footnotesize LinsatConstraint \\
            \scriptsize\color{darkgrey} expr: LinsatExpr\\
            \scriptsize\color{darkgrey} rhs: LinsatRhs
            };
            \node[box, below=of linsatexpr] (linsatvar) {
            \footnotesize LinsatVar \\
            \scriptsize\color{darkgrey} id: int\\
            \scriptsize\color{darkgrey} name: str
            };
            \node[box, below=of linsatconstraint] (linsatrhs) {
            \footnotesize LinsatRhs \\
            \scriptsize\color{darkgrey} weights: dict[\\
            \scriptsize\color{darkgrey} \,\,\,\, sage.GFElement, int\\
            \scriptsize\color{darkgrey} ]
            };
            \draw[-stealth, thick] (maxlinsat) -- (linsatconstraint);
            \draw[-stealth, thick] (linsatconstraint) -- (linsatexpr);
            \draw[-stealth, thick] (linsatexpr) -- (linsatvar);
            \draw[-stealth, thick] (linsatconstraint) -- (linsatrhs);

            \begin{scope}[on background layer]
\node[
    draw=none,
    fill=lfd4lighter,
    fit=(maxlinsat)(linsatexpr)(linsatrhs)(relation.south west -| linsatrhs),
    label={[anchor=north,text=lfd4]north:\small Lower abstraction level: LINSAT constraints},
    yshift=4pt,
    inner ysep=12pt,       
    inner xsep=8pt,
    rounded corners=5pt
] (group1) {};
\node[
    draw=none,
    fill=lfd7lighter,
    fit=(maxconstraintsat)(intexpr)(intconstraint)(relation)(intvar),
    label={[anchor=north,text=lfd7]north:\small Higher abstraction level: Boolean/integer constraints/objectives},
    yshift=4pt,
    inner ysep=12pt,       
    inner xsep=8pt,
    rounded corners=5pt
] (group2) {};
\end{scope}
        \end{tikzpicture}
    \caption{Simplified class structure for the two abstraction levels of \framework: \texttt{MaxLinSat} for directly describing Max-LINSAT constraints and \texttt{MaxConstraintSat} for describing integer constraints and objectives as well as Boolean constraints.}
    \label{fig:class-diagram}
\end{figure*}

\section{Handling Limitations of DQI} \label{sec:weighted-max-linsat}
DQI suffers from two limitations with regards to the problem formulations described in \autoref{sec:problem-encoding}:
Firstly, it does not support weighted constraints;
secondly, it requires all sets $F_i$ to be of the same size.
Both limitations arise because the first step of DQI involves preparing a trial state that is a superposition of all possible error vectors of weight at most $\ell$.
This trial state is symmetric with respect to different error positions.
Due to this symmetry, DQI can only prepare the DQI state \eqref{eq:dqi-state} exactly if the objective function $f$ is also unbiased with respect to different constraints.
This results in symmetry requirements for the problem instance that are only met when all constraints have the same weight and all sets $F_i$ have the same size.

It is plausible that both symmetry requirements can be lifted by using an asymmetric trial state:
In a recent work on a generalised version of DQI, Bu~\etal \cite{Bu2026} showed that the DQI state can be prepared efficiently for Weighted Max-XORSAT if $\ell < d_\text{min} / 2$.
This is done by using an asymmetric matrix product state as the trial state, which has bond dimension $\ell + 1$ and can thus be prepared efficiently in time polynomial in $m$ and $\ell$.
This approach may be generalisable for any field $\mathbb{F}_q$ and for sets $F_i$ of different size, but this is not guaranteed.
Additionally, computing the expected number of satisfied constraints achieved by DQI is likely less efficient than in the symmetric case, potentially allowing for less precise performance evaluation.
Thus, our framework needs to address these symmetry limitations.

The most straightforward~(and, in general, only viable) option to handle weighted constraints is to duplicate them such that the number of copies of a constraint is proportional to its weight.
However, this creates two problems in turn:
First, large weights significantly increases the number of constraints.
Second, as explained in \autoref{sec:suitable-max-linsat-instances}, duplicated constraints heavily limit the approximation capabilities of DQI.
To mitigate the first issue, we can divide the weights by a common factor if it is shared by all constraint weights, reducing the number of required duplicates.
If the weights do not share a common factor, we can round them to a multiple of an integer $d$.
However, doing so alters the problem we are trying to solve, generally decreasing the solution quality of DQI.
Rounding thus works best if weights do not vary a lot between constraints.
To mitigate the second problem of duplicate constraints resulting in a minimum distance of $2$, we can employ dependency reduction techniques like the the gadget defined in \autoref{thm:gadget}.
If we repeat a constraint $w_k$ times, we can resolve the dependency by adding this gadget to $w_k - 1$ of the $w_k$ duplicates.

Finally, we need to handle $F_i$ sets of different sizes.
This can also be achieved via duplicate constraints along with dependency reduction techniques.
First, we find the greatest common divisor $d$ of all set sizes and split each $F_i$ into sets of size $d$.
This works especially well for integer constraints as outlined \autoref{sec:equality-inequality-constraints}.
For these types of constraints, we have some flexibility when choosing the sets $F_i$, since the field order does not usually fit the constraint's range exactly.
We can therefore increase the greatest common divisor by deliberately including values in the sets $F_i$ that should not be attainable without violating the variables' range constraints in such a way that the sizes of the $F_i$ have large common factors.
In some cases, this approach might still not be viable.
It thus remains an interesting open question whether the same-size restriction can be lifted, as was done for the weight restriction in \cite{Bu2026}.

\section{\framework} \label{sec:framework}
The main contribution of our work is \framework, an open source Python library to express various types of domain objectives and constraints, which are then converted into weighted or unweighted Max-LINSAT instances via the transformations depicted in \autoref{fig:transformations}.
\framework then estimates the approximation performance of both DQI and a series of classical algorithms on these transformed instances.

\begin{figure*}
    \centering
    \begin{tikzpicture}[
      node distance=0.5cm and 0.2cm,
      every node/.style={
          draw=lfd3,
          rounded corners=5pt,
          rectangle,
          inner xsep=4pt,
          inner ysep=-3pt
        }
    ]
        
    \node[draw=lfd4, thick] (colourable) {
        \begin{minipage}{0.26\textwidth}
            \begin{lstlisting}
# Max. 3-colourable subgraph
G = nx.Graph(...)
prob = MaxLinSat(GF(3))
x = [
  prob.new_var(f"x_{i}")
  for i in range(n)
]
for u, v in G.edges:
  prob.add_constraint(
    x[u] != x[v]
  )
            \end{lstlisting}
        \end{minipage}
    };
    \node[draw=lfd7, thick, left=0.6cm of colourable.north west, anchor=north east] (vertex-cover) {
        \begin{minipage}{0.28\textwidth}
            \begin{lstlisting}
# Min. vertex cover
G = nx.Graph(...)
prob = MaxConstraintSat()
x = [
  prob.new_binary_var(f"x_{i}")
  for i in range(n)
]
     
prob.add_objective(
  sum(x_i for x_i in x),
  minimize=True,
)
for u, v in G.edges:
  prob.add_boolean_constraint(
    x[u] | x[v], weight=2
  )
            \end{lstlisting}
        \end{minipage}
    };
    \node[draw=lfd7, thick, left=of vertex-cover] (knapsack) {
        \begin{minipage}{0.35\textwidth}
            \begin{lstlisting}
# Knapsack
items = [Item(w=3, val=5), ...]
prob = MaxConstraintSat()
x = [
  prob.new_binary_var(f"x_{i}")
  for i in range(len(items))
]
prob.add_constraint(
  sum(
    x_i * i.w for x_i,i in zip(x,items)
  ) <= W,
  weight=2 * sum(i.val for i in items)
)
prob.add_objective(sum(
  x_i * i.val for x_i,i in zip(x,items)
))
            \end{lstlisting}
        \end{minipage}
    };

    \coordinate (center) at ($(knapsack.west)!0.5!(colourable.east)$);
    
    \node[anchor=north, thick, yshift=-0.5cm] (solver) at (center |- vertex-cover.south) {
        \begin{minipage}{0.45\textwidth}
            \begin{lstlisting}
print(SimAnnealSolver(prob).get_solution_quality())
dqi = Dqi(prob)
for l in range(0, n):
  print(dqi.estimate_solution_quality(l))
            \end{lstlisting}
        \end{minipage}
    };

    \tikzset{
        arrow/.style={
            thick,
            -stealth,
            draw=lfd3,
            rounded corners=3pt
        }
    }

    \draw[arrow] (colourable) |- (solver);
    \draw[arrow] (vertex-cover.south) -- (solver.north -| vertex-cover.south);
    \draw[arrow] (knapsack) |- (solver);
    \end{tikzpicture}
    \caption{
        Problem encoding and solution quality estimation in \framework:
        User-defined problem specifications supports two abstraction levels:
        in this example, higher-level \texttt{MaxConstraintSat} is used for \emph{knapsack} and \emph{minimum vertex cover};
        lower-level abstraction \texttt{MaxLinSat} is used for \emph{maximum 3-colourable subgraph}.
        DQI and classical solvers operate on both.
        For this, \texttt{MaxConstraintSat} instances are automatically transformed into Max-LINSAT.
    }
    \label{fig:code-problems}
\end{figure*}
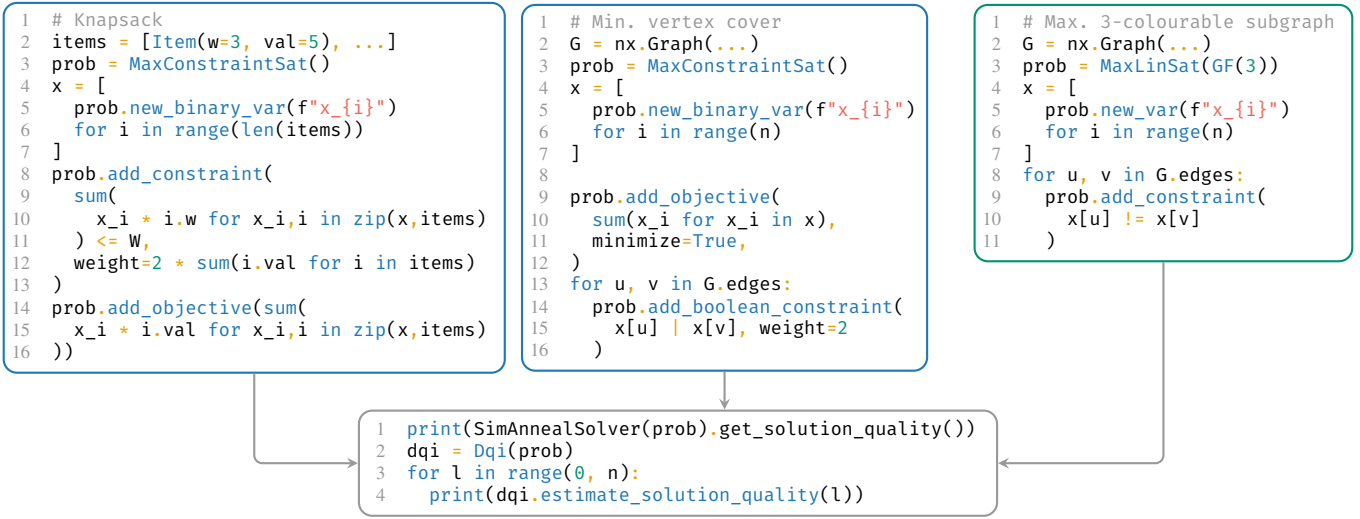

\subsection{Abstractions for Max-LINSAT Constraints}
The general software architecture of our framework is visualised in \autoref{fig:class-diagram}.
\framework provides two main ways for users to describe constraints, which differ in their level of abstraction.
Using the lower-level \texttt{MaxLinSat} class, users can encode $=$, $\neq$ and $\in$ Max-LINSAT constraints along with their weights directly, allowing for a more manual control of the precise problem formulations.
For this, users can choose the finite field $\mathbb{F}_q$ they want to operate in.
For the field operations we use \emph{Sage} \cite{Sage2025}, which provides implementations for arbitrary finite fields via a unified interface to several efficient backends (\eg Refs.~\cite{Pari,Ntl,Givaro}).
Variables are created via \texttt{MaxLinSat.new\_var}.
They always represent elements in $\mathbb{F}_q$.
Using Python's operator overloading capabilities, constraints can then be described via intuitive syntax as shown in \autoref{fig:code-problems}.
Equalities and inequalities are rearranged automatically to conform to the Max-LINSAT formalism (\autoref{def:max-linsat}).
For typical constraint patterns beyond $=$, $\neq$ and $\in$ constraints, such as Boolean constraints, \texttt{MaxLinSat} supports so-called \emph{gadgets}.
A gadget can be thought of as a function mapping a list of $n$ variables to a set of $m$ constraints.
It is defined by a miniature Max-LINSAT instance, similar to the one outlined in \eqref{eq:and}.
Users can define their own reusable gadgets.
We also provide a small library of gadgets for various Boolean constraints as well as a mechanism to generate a gadget based on a given truth table.
This mechanism relies on exhaustively searching all possible sets of constraints.
Using it is thus only feasible for small gadgets and field orders. 

\texttt{MaxLinSat} automatically merges constraints with the same left-hand side into a single constraint.
The right-hand side of a constraint is represented as a hash map mapping field elements to weights.
For example, merging the constraint $x + y = 0$ with weight $1$ and the constraint $x + y \in \{0, 1\}$ with weight $2$ results in a constraint that internally represents the right-hand side as the hash map \texttt{\{0: 3, 1: 2\}}.
Here, the entry \texttt{0: 3} refers to the fact that $x + y$ evaluating to $0$ would contribute a total weight of $3$ to the Max-LINSAT objective function.
Only integer weights are supported.
Each \texttt{MaxInstance} instance can thus be interpreted as both a weighted and, via constraint duplication, an unweighted Max-LINSAT problem.
This allows solvers that do not support weighted constraints, such as standard DQI, to be applied to any \texttt{MaxLinSat} instance.

\subsection{Abstractions for Integer and Boolean Constraints}
The second abstraction layer provided by \framework is represented by the class \texttt{MaxConstraintSat}.
It allows for the specification of various types of (potentially weighted) Boolean, integer and modular constraints.
The resulting constraint problem can then be transformed into a \texttt{MaxLinSat} instance, via a call to \texttt{to\_max\_lin\_sat}.
Users can create variables with the \texttt{new\_var} method.
As shown in \autoref{fig:class-diagram}, these variables can be combined with other variables and scalar constants, represented by the Fraction class, to form expressions (\texttt{IntExpr}) via addition and multiplication.
The set of representable expressions is therefore equal to the set of multi-variate polynomials.
An \texttt{IntExpr} can be used directly as objective to be maximised or minimised, or, along with a relation ($=, \neq, \leq, <, \geq, >$), as part of an \texttt{IntConstraint}.
When describing a modular constraint, only $=$ and $\neq$ are supported since there is no well-defined notion of order in $\mathbb{Z}_n$.
By convention, the right-hand side of any \texttt{IntConstraint} is $0$ and expressions are automatically rearranged accordingly.

\texttt{MaxConstraintSat} has no explicit notion of a binary variable.
At creation, the value range of each \texttt{IntVar} is specified via integer lower and upper bounds.
A variable is considered binary if its lower bound is $0$ and its upper bound is $1$.
In addition to \texttt{+}, \texttt{-} and \texttt{*}, binary variables also support the Python Boolean operators \texttt{\~}, \texttt{\&}, \texttt{|} and \texttt{\^}, representing NOT, AND, OR and XOR respectively.
These Boolean operators are automatically converted to addition and multiplication.
For example, $a \vee b$ is equivalent to $a + b - a b$.
This means that, internally, a Boolean constraint is handled identically to an integer expressions and is therefore added as an objective and not as a constraint since it includes no relation such as $=$.

The transformation of \texttt{MaxConstraintSat} to \texttt{MaxLinSat} happens in multiple passes.
In the first pass, polynomial objectives and constraints as well as linear objectives are transformed into linear constraints, as outlined in Sections \ref{sec:linear-objectives} and \ref{sec:qubos-pubos}.
As explained in \autoref{sec:qubos-pubos}, a pseudo-Boolean function over $k$ variables can lead to up to $2^k$ Max-LINSAT constraints.
To combat this, the transformation splits up higher order terms introducing auxiliary variables.
This reduces monomial degrees, and thus the number of Max-LINSAT constraints required.
By default, terms of degree at least four are split up, but this behaviour can be changed by the user.
Auxiliary variables are also introduced to represent higher order terms within constraints.

After the first pass, only linear constraints remain.
These are transformed into modular constraints.
If needed, constraints are scaled such that they only contain integer coefficients.
Then, the smallest prime $p$ is determined so all constraints can be encoded modulo that prime and the integer constraints are converted into constraints over $\mathbb{F}_p$.
Next, for each variable $x_i$ with lower bound $l_i$ and upper bound $u_i$ with $u_i - l_i + 1 < p$, a constraint $x_i \in \{l_i, \dots, u_i\}$ is added.
Finally, we end up with three types of constraints: user-defined constraints, auxiliary variable constraints for higher-order terms and variable range constraints.
The latter two types are automatically assigned weights that are large enough to give them priority over user-defined constraints, unless the user specifies otherwise.

\subsection{Solvers and Decoders}
Our software framework implements a DQI solver and four classical solvers for Max-LINSAT instances.
Actually executing DQI on a real or simulated quantum computer is impossible with current hardware, for industrial-scale instances.
Instead of performing the algorithm, the DQI solver therefore computes the expected solution quality in terms of the expected number of satisfied constraints using the closed expressions derived by Jordan \etal~\cite{Jordan2025}.
The user can choose (a) the degree $\ell$ of the polynomial $P$ (which determines how strongly the final superposition prepared by DQI is biased towards the optimal solution), and (b)
the classical decoder used during syndrome decoding.
We provide four general-purpose decoding algorithms (belief propagation, information set decoding, lookup table syndrome decoding and nearest-neighbour decoding) as well as an interface to add other decoders. We use the \texttt{ldpc} package \cite{Ldpc2022} for belief propagation and Sage \cite{Sage2025} for the other decoders.
Note that the lookup table syndrome decoder and the nearest-neighbour decoder correct all theoretically correctable errors but they have exponential runtime in $\ell$.
Nevertheless, they can efficiently provide provable upper bounds on problem instances with small minimum distance making them useful for ruling out potential problem candidates for DQI.
If the user does not specify $\ell$ or the decoder, both are selected automatically depending on the order of the field $\mathbb{F}_q$ and the approximate minimum distance of the linear code $\Code$ dual to the given Max-LINSAT instance.
If not otherwise specified, \framework also selects the optimal coefficients of $P$ automatically by solving an eigenvalue problem as described in~\cite{Jordan2025}.

If $\ell$ is smaller than half the minimum distance of $\Code$, we can efficiently compute the expected number of satisfied constraints exactly.
Otherwise, computing this quantity either takes exponential time or can only be done approximately by sampling from the distribution of possible error vectors.
In \framework, the user can select if they want an exact or approximated value and, in the latter case, how many errors should be sampled.
This give a more accurate approximation at the expense of compute time.

Along with the DQI solver, we provide interfaces for four classical methods:
a brute force solver for finding the optimal solution,
the constraint programming solver from Googles OR-Tools \cite{Ortools2025},
a simulated annealing solver using the \texttt{simanneal} package \cite{Simanneal2019} and,
finally, Prange's algorithm.
Prange's algorithm finds an approximate solution by randomly selecting a set of $n$ constraints, whose left-hand sides are linearly independent.
The resulting system of equations is then solved exactly using Gaussian elimination.
If we assume the $m - n$ remaining constraints to be random and independent from the $n$ selected, this satisfies $n + (m - n) / q$ equations in total, on average.
Prange's algorithm is noteworthy as it is the currently best known classical algorithm for the problem for which \cite{Jordan2025} showed apparent quantum advantage with DQI.

\section{Related Work} \label{sec:related-work}
DQI was proposed by Jordan \etal~\cite{Jordan2025} as a quantum algorithm that finds approximate solutions to a combinatorial optimisation problem via so-called \emph{Regev reductions}.
This type of reduction is based on the relationship via a Fourier transform between decoding problems in two codes dual to each other.
Regev reductions have has been studied before DQI, both framed in the context of lattice problems \cite{Aharonov2003,Aharonov2005,Regev2009} and coding theory \cite{Chen2022,Yamakawa2024,Debris2024}.
There have been several works building upon the original DQI paper \cite{Chailloux2025,Khattar2025,Schmidhuber2025,Gu2025,Marwaha2025}.
In particular, we want to highlight the works of Anschuetz \etal~\cite{Anschuetz2025} and Kramer \etal~\cite{Kramer2026}, who provide
compelling evidence indicating that quantum advantage with DQI is not achievable on unstructured instances.
We also want to highlight the work of Parekh~\cite{Parekh2025}, who ruled out quantum advantage for Max-Cut with DQI when using perfect decoding algorithms.
These results underscore the importance of identifying suitable problem instances for utilising DQI in industrial applications.
DQI has been successfully applied to algebraic problems derived from known good error-correcting codes \cite{Jordan2025, Khattar2025, Gu2025, Chailloux2025, Hillel2025}.
However, to the best of our knowledge, Sabater \etal~\cite{Sabater2025} published the first investigation of DQI on an industrial use case.
They apply DQI to a binary integer linear program and but find that, for their problem formulation, DQI is not competitive against state-of-the-art classical approaches.

Software-aided transformations of domain problems~\cite{Yue2023,Carbonelli2024} into mathematical formalisms have been broadly studied for other problem formulations:
For instance, interfaces to encode various types of constraints are provided by many state-of-the-art mixed-integer linear programming solvers including Gurobi, Google OR-Tools or the SCIP Optimisation Suite~\cite{Gurobi2026,Ortools2025,Hojny2025}.
In addition, there has been extensive research on the development of domain-specific languages for problem encoding \cite{Nethercote2007,Dunning2017,Nicholson2017}.
In quantum optimisation, QUBOs the most widely studied problem formulation.
Multiple efforts provide tools to encode various constraints as QUBOs, which, in turn, can then be used in quantum approaches such as QAOA or quantum annealing~\cite{mauerer:26:qdata,Zaman2021,Lobe2023,Quetschlich2023,Rovara2024}.

\section{Conclusion and Outlook} \label{sec:conclusion}
We presented \framework, an open source software framework that allows users to study the performance of the DQI algorithm on broad classes of industrial combinatorial optimisation problems.
We demonstrated that high-level user-specified descriptions of predominant constraint and objective types can be encoded into problem formulations directly compatible with DQI via a series of transformations.
Based on the limitations of DQI, we derived a series of guidelines on how to determine suitable problem instances and on which types of constraints can degrade the performance of the algorithm.

Our work shows that many transformations introduce inefficiencies in the form of auxiliary variables, an increased number of constraints or linear dependencies in the constraint matrix.
These inefficiencies can limit the performance of DQI, sometimes significantly.
We demonstrate that, in principle, mitigating some of these inefficiencies is possible to a certain extent.
To significantly expand the class of problems where DQI is applicable, further research into more efficient transformations and more sophisticated dependency mitigation techniques is required.
Developing generalisations of DQI for weighted constraints or heterogeneous right-hand side sets could also extend its applicability considerably by allowing for more efficient encoding of many practical use cases.
We envision \framework as a starting point for a community-driven standard that provides a unified interface for a wide range of possible transformation and encoding strategies.

\newcommand{\WM}{\censor{WM}\xspace}
\newcommand{\hta}{\censor{High-Tech Agenda Bavaria}}
\small{
\textbf{Acknowledgements}
This work was supported by the German Federal Ministry of Research, Technology and Space (BMFTR), funding program ‘Research Program Quantum Systems’, grant number 13N16092. We acknowledge partial support
by the German Research Foundation, grant MA 9739/1-1,
by the European
Union (Project Reference 101083427) and the European Funds for Regional Development (EFRE) (Project
Reference 20-3092.10-THD-105).
\WM acknowledges support by the \hta.
}

\newpage\bibliographystyle{IEEEtran}
\bibliography{references}

\end{document}